\documentclass[11pt]{article}
\usepackage{comment}
\usepackage{etex}
\usepackage{booktabs}

\usepackage{graphicx}
\usepackage{amsthm,amsmath}
\usepackage{amssymb}
\usepackage[ruled,vlined]{algorithm2e}
\usepackage{float}
\usepackage{paralist}
\usepackage{xcolor}
\usepackage{color}
\usepackage{url}
\usepackage{booktabs}
\renewcommand\vec{\mathbf}

\input prepictex
\input postpictex
\input pictexwd

\usepackage[affil-it]{authblk}
\usepackage[round]{natbib}

\usepackage{todonotes}

\newtheorem{theorem}{Theorem}
\newtheorem{corollary}[theorem]{Corollary}
\newtheorem{lemma}[theorem]{Lemma}
\newtheorem{proposition}{Proposition}
\newtheorem{example}{Example}
\newtheorem{definition}{Definition}
\newtheorem{remark}{Remark}

\newcommand{\calR}{\mathcal{R}}
\newcommand{\calL}{\mathcal{L}}
\newcommand{\calA}{\mathit{Alg}}
\newcommand{\tp}{\mathrm{top}}
\newcommand{\scr}{\mathrm{sc}}
\newcommand{\GS}{\mathcal{GS}}
\newcommand{\GSP}{\mathcal{GSP}}

\title{\bf Cognitive Hierarchy and Voting Manipulation\thanks{Some of the results in this paper 
were presented at the 24th International Joint Conference on Artificial Intelligence (IJCAI-15) 
\citep{ElkindEtAlIJCAI2015}. We are grateful to the IJCAI-15 reviewers, 
as well as to the audiences of the Workshop on
Iterative Voting and Voting Games in Padova in 2014, 
the 5th International Workshop on Computational Social Choice (COMSOC-2014) in Pittsburgh, 
and to the participants of seminars in Toulouse in 2014 and in Dagstuhl in 2015 
for their useful feedback.}}

\author{Edith Elkind}
\affil{University of Oxford, elkind@cs.ox.ac.uk}

\author{Umberto Grandi}
\affil{University of Toulouse, umberto.grandi@irit.fr}

\author{Francesca Rossi}
\affil{University of Padova, frossi@math.unipd.it}

\author{Arkadii Slinko}
\affil{The University of Auckland, a.slinko@auckland.ac.nz}

\begin{document}

\maketitle

\begin{abstract}
\noindent
By the Gibbard--Satterthwaite theorem, every reasonable voting rule for three or more alternatives
is susceptible to manipulation: there exist elections where one or more voters can change the election
outcome in their favour by unilaterally modifying their vote.
When a given election admits several such voters, strategic voting becomes
a game among potential manipulators: a manipulative vote that leads to a better outcome
when other voters are truthful may lead to disastrous results when other voters
choose to manipulate as well. We consider this situation from the perspective of a boundedly rational voter, 
and use the cognitive hierarchy framework \citep{camerer2004cognitive} to identify good strategies.
We then investigate the associated algorithmic questions
under the $k$-approval voting rule, $k\ge 1$. We obtain positive algorithmic results
for $k=1, 2$ and NP- and coNP-hardness results for $k\ge 4$.
\end{abstract}

\section{Introduction}\label{sec:intro}
Imagine that you and your friends are choosing a restaurant to go to for dinner.
Everybody is asked to name their two most preferred cuisines, 
and the cuisine named most frequently will be selected (this voting rule is known as 2-approval). 
Your favourite cuisine is Japanese and your second most preferred cuisine is Indian.
Indian is quite popular among your friends and you know that if you name it among your favourite two cuisines, it
will be selected. On the other hand, you also know that only a few of your friends like Chinese food.
Will you vote for Japanese and Chinese to give Japanese cuisine a chance?

This example illustrates that group decision-making is a complex process that represents an
aggregation of individual preferences. Individual decision-makers would like to influence the final decision in a
way that is beneficial to them, and hence they may be strategic in communicating their individual choices.
Moreover, it is essentially impossible to eliminate strategic behavior by changing the voting rule:
the groundbreaking result of \citet{gib:j:gs} and \citet{sat:j:gs} states that, under any 
onto and non-dictatorial social choice rule, there exist situations where a voter can achieve a better outcome by 
casting a strategic vote rather than the sincere one, provided that everyone else votes sincerely;
in what follows, we will call such voters {\em Gibbard--Satterthwaite (GS) manipulators}. 

The Gibbard--Satterthwaite theorem alerts us that strategic behavior of voters cannot be ignored, 
but it does not tell us under which circumstances it actually happens. Of course, if there is just a 
single GS-manipulator at a given profile, and he\footnote{We use `he' to refer to voters and `she' to refer to candidates.}
is fully aware of other voters' preferences, it is rational for him to manipulate. 
However, even in this case this voter may prefer to vote truthfully, simply because 
he assigns a high value to communicating his true preferences; such voters are called {\em ideological}.
Moreover, if there are two or more GS-manipulators, it is no longer easy for them to make up their
mind in favour of manipulation: while the Gibbard--Satterthwaite
theorem tells us that each of these voters would benefit from voting strategically 
assuming that all other voters remain truthful, it does not offer any predictions if
several voters may be able to manipulate simultaneously. 
The issues faced by GS-manipulators in this case are illustrated by the following example.

\begin{example}\label{ex1}
{\em
Suppose four people are to choose among three alternatives by means of 2-approval,
with ties broken according to the order $a>b>c$.
Let the profile of sincere preferences be as in Table~\ref{table:ex1}.
There are two voters who prefer $b$ to $c$ to $a$, one voter who prefers $a$ to $c$ to $b$,
and one voter who prefers $c$ to $b$ to $a$.
If everyone votes sincerely, then $c$ gets 4 points,
$b$ gets 3 points and $a$ gets 1 point, so $c$ is elected. 
Voters 1 and 2 are Gibbard--Satterthwaite manipulators. Each of them  can make $b$ the winner by voting $bac$,
ceteris paribus.  Let us consider this game from the first voter's perspective, 
assuming that he is strategic; let $A_i$ denote the strategy set of voter $i$, $i=1, 2, 3, 4$. 
The strategy set of voter 1 can then be assumed to be $A_1=\{bca, bac\}$
(clearly, under 2-approval $cba$ is indistinguishable from $bca$, $abc$ is indistinguishable from $bac$, 
and the two votes that do not rank $b$ first are less useful than either $bca$ and $bac$).
Voter 1 has a good reason to believe that voters 3 and 4 will vote sincerely, 
as voter 3 cannot achieve an outcome that he would prefer to the current outcome and voter 4 is fully satisfied.

\begin{table}[h]
\begin{center}
\begin{tabular}{c|c|c|c}
voter 1&voter 2&voter 3&voter 4\\
\midrule
$b$&$b$&$a$&$c$\\
$c$&$c$&$c$&$b$\\
$a$&$a$&$b$&$a$\\
\end{tabular}
\caption{A preference profile. The most preferred candidates are on top, followed by the less preferred candidates in a complete ranking.
}\label{table:ex1}
\end{center}
\end{table}


{\bf Case 1.} 
If voter 1 believes that voter 2 is ideological, then he is analysing the game where $A_2=\{bca\}$, $A_3=\{acb\}$ 
and $A_4=\{cba\}$. In this case he just votes $bac$ and expects $b$ to become the winner. 

{\bf Case 2.} 
Suppose now that voter 1 believes that voter 2 is also strategic. Now voter 1 has to analyse the game with 
$A_1=A_2=\{bca, bac\}$, $A_3=\{acb\}$ and $A_4=\{cba\}$. If either one of the strategic players---voter 1 or voter 
2---manipulates and another stays sincere, $b$ will be the winner. However, if they both manipulate, their worst 
alternative $a$ will become the winner.  
Thus, in this case voter 1's manipulative strategy does not dominate his sincere vote, 
and if voter~1 is risk-averse, he should refrain from manipulating.
}
\end{example}

A popular approach (see Section~\ref{sec:related}) is to view voting as a strategic game among the voters, 
and use various game-theoretic solution concepts to predict the outcomes. 
The most common such concept is Nash equilibrium, which 
is defined as a combination of strategies, one for each player, 
such that each player's strategy is a best response to other players' strategies. 
In these terms, the Gibbard--Satterthwaite theorem says that under every reasonable voting rule 
there are situations where truthful voting is not a Nash equilibrium.
As a further illustration, the game analysed in Example~\ref{ex1} (Case~2) has two Nash equilibria:
in the first one, voter 1 manipulates and voter 2 remains truthful, and in the second one the roles
are switched. 

However, the principle that players can always be expected to choose equilibrium strategies is 
not universally applicable. Specifically, if players have enough experience 
with the game in question  (or with similar games), both theory and experimental 
results suggest that players are often able to learn equilibrium strategies \citep{fudenberg1998theory}. However, it is 
also well-known since the early work of \cite{shapley64} that learning dynamics may fail to converge to an equilibrium.
Moreover, in many applications---and voting is one of them---players' interactions have only imperfect precedents, 
or none at all so no learning is possible. If equilibrium is justified in such applications, it must be via 
strategic thinking of players rather than learning. However, in some games the required reasoning is too complex for 
such a justification of equilibrium to be behaviourally plausible 
\citep{harsanyi1988general,brandenburger1992knowledge}. This is fully applicable to voting, where such reasoning, 
beyond very simple profiles, is impossible because of the number of voters involved.

In fact, a number of recent experimental and empirical studies suggest that players' responses in strategic 
situations often deviate systematically from equilibrium strategies, and are better explained by the structural nonequilibrium 
level-$k$ \citep{nagel1995unraveling,stahl1994experimental} or cognitive hierarchy (CH) models 
\citep{camerer2004cognitive}; see also a survey by \cite{crawford2013structural}. 
In a level-$k$ model players anchor their 
beliefs in a non-strategic initial assessment of others' likely responses to the game.
Non-strategic players are said to be level-0 players. 
Level~1 players believe that all other players are at level~0, 
and they give their best response on the basis of this belief. 
Level~2 players assume that all other players belong to level~1, and, more generally, 
players at level $k$ give their best response assuming that all other players are at level $k-1$. 
The cognitive hierarchy model is similar, but with an essential difference: 
in this model players of level $k$ respond to a mixture of types from level $0$ to level $k-1$. 
It is frequently assumed that other players' levels are drawn from a Poisson distribution. 
Some further approaches based on similar ideas are surveyed by \citet{aaai/WrightL10}. 

The aim of our work is to explore the applicability of these models to voting games.
We believe that specifics of voting, and, in particular, the heterogeneity of types
of voters in real electorates, make the cognitive hierarchy framework more appropriate 
for our purposes. In more detail, an important feature that distinguishes voting from many other applications 
of both level-$k$ and CH models is the role of level-0 players. Specifically,  
level-0 (non-strategic) players are typically assumed to choose their strategy at random, 
and this type practically does not appear in real games at all. 
In contrast, in applications to voting it is natural to associate level-0 players 
with ideological voters, who have a significant presence in real elections. 
For instance, in the famous Florida vote (2000), where Bush won over Gore by just 537 votes,  
97,488 Nader supporters voted for Nader---even though in such a close election every strategic voter 
should have voted either for Gore or for Bush (and an overwhelming majority of Nader supporters 
preferred Gore to Bush). However, in the level-$k$ analysis voters of level~2 assume 
that all other voters have level 1, i.e., level-$k$ models cannot be used
to accommodate ideological voters. We therefore focus on the cognitive hierarchy approach.
Moreover, we limit ourselves to considering the first three levels of the hierarchy
(i.e., level-0, level-1, and level-2 players), as it seems plausible that very few voters
are capable of higher-level reasoning (see the survey by \citet{crawford2013structural}
for some evidence in support of this assumption). 

Adapting the cognitive hierarchy framework to voting games is not a trivial task.
First, it does not make sense to assume that voters' levels follow a specific distribution.
Second, in the standard model of social choice voters' preferences over alternatives
are ordinal rather than cardinal. The combination of these two factors means that, 
in general, for voters at level~2 or higher their best response may not be well-defined.
We therefore choose to focus on strategies that are not weakly dominated according
to the voter's beliefs. We present our formal definitions and the reasoning that
justifies them in Section~\ref{sec:model}.

To develop a better understanding of the resulting model, we instantiate it for a specific family
of voting rules, namely, $k$-approval with $k\ge 1$. We develop a classification of \mbox{level-1} 
strategies under $k$-approval and clarify the relationship between level-1 reasoning and the predictions of the 
Gibbard--Satterthwaite theorem (Section~\ref{sec:classification}). We then switch our attention to level-2 strategies, 
and, in particular, to the complexity of computing such strategies.
For $k$-approval with $k=1$ (i.e., the classic Plurality rule) we describe an efficient algorithm 
that decides whether a given strategy weakly dominates another strategy; as a corollary of this result, 
we conclude that under the Plurality rule level-2 strategies can be efficiently computed 
and efficiently recognised (Section~\ref{sec:plu}). We obtain a similar result for 2-approval under 
an additional assumption of {\em minimality}
(Section~\ref{sec:2app}). Briefly, this assumption means that the level-2 player expects all level-1 players
to manipulate by making as few changes to their votes as possible. However, for larger values of $k$
finding level-2 strategies becomes computationally challenging: we show that this problem
is NP-hard for $k$-approval with $k\ge 4$ (Section~\ref{sec:hard}). As the problem of finding a level-1 strategy
under $k$-approval is computationally easy for any value of $k\ge 1$ (this follows immediately
by combining our characterization of level-1 strategies with the classic results of \citet{btt89}),
this demonstrates that higher levels of voters' sophistication come with a price tag in terms
of algorithmic complexity.


\subsection{Related work}\label{sec:related}
There is a substantial body of research in social choice theory and in political science 
that models non-truthful voting as a strategic interaction,
with a strong focus on Plurality voting;
this line of work dates back to \citet{far:b:voting} and includes important contributions 
by \citet{Cain1978}, \citet{fed-sen-wri:j:entry} and \citet{Cox1997}, to name a few. 

More recently, voting games and their equilibria have also received 
a considerable amount of attention from computer science researchers, with
a variety of approaches used to eliminate counterintuitive Nash equilibria.
For instance, some authors assume that voters have a slight preference for abstaining
or for voting truthfully when they are not pivotal
\citep{bat:j:abstentions,dut-sen:j:nash,des-elk:c:eq,tho-lev-ley:c:empirical,obr-mar-tho:c:truth-biased,emos15,OLMRR15}.
Other works consider refinements of Nash equilibrium, such as subgame-perfect Nash equilibrium 
\citep{des-elk:c:eq,xia-con:c:spne}, strong equilibrium \citep{mes-pol:j:strong} 
or trembling-hand equilibrium \citep{OREPJ16},
or model the reasoning of voters who have incomplete or imperfect information about each others'
preferences \citep{mye-web:j:voting,Myatt2007,MLR14}.
Dominance-based solution concepts have been investigated as well 
\citep{mou:j:dominance,dhi-loc:j:dominance,bue-dhi-vid:j:domsolv,dellis2010weak,MLR14}, 
albeit from a non-computational perspective.
All the aforementioned papers do not impose any restrictions on the voters' reasoning ability, 
de facto assuming that they are fully rational. Boundedly rational voters are considered
by \citet{GHRS}; however, their work focuses on strategic interactions among 
Gibbard--Satterthwaite manipulators, and studies conditions that ensure existence of 
pure strategy Nash equilibria in the resulting games.
In contrast, in this paper we go further and formally define the degree of voters' rationality  
by using the cognitive hierarchy approach. 

Level-$k$ models and the cognitive hierarchy framework have been long used to model 
a variety of strategic interactions;
we refer the reader to the survey of \citet{crawford2013structural}. Nevertheless, to the best
of our knowledge, ours is the first attempt to apply these ideas in the context of voting.

A  topic closely related to voting games is voting dynamics, where players change their votes one by one in response to the current outcome
\citep{mei-pol:c:convergence,rei-end:c:polls,rey-wil:c:bestreply,OMPRJ15,EOPR16,lev-ros:j:iterative,KSLR17}; 
see also a survey by \citet{Meir-trends}.  
However, this line of work assumes the voters to be myopic.

Our work can also be seen as an extension of the model of safe strategic voting proposed by \citet{safe2}. 
However, unlike us, Slinko
and White focus on a subset of GS-manipulators who (a) all have identical preferences
and (b) choose between truthtelling and using a specific manipulative vote,  
and on the existence of a weakly dominant strategic vote in this setting
(such votes are called {\em safe strategic votes}). In contrast, our decision-maker takes into account that
manipulators may have diverse preferences and have strategy sets that contain more than one strategic vote.
It is therefore not surprising that computing safe strategic votes is easier than finding level-2 strategies:
\citet{safe-sagt} show that safe strategic votes with respect to $k$-approval can be computed efficiently
for every $k\ge 1$, whereas we obtain hardness results for $k\ge 4$. 

One of our contributions is a classification of manipulative votes under $k$-approval with 
lexicographic tie-breaking. \citet{peters2012manipulability} propose a similar classification for several 
approval-based voting rules. However, they view $k$-approval as a non-resolute voting rule, and therefore 
their results do not apply in our setting.

\paragraph{Paper outline.} The paper is organised as follows.
We introduce the basic terminology and definitions in Section~\ref{sec:prelim}.
Section~\ref{sec:model} presents the adaptation of the cognitive hierarchy framework to the setting of voting games.
We then focus on the study of $k$-approval. Section~\ref{sec:classification} 
describes the structure of level-1 strategies under $k$-approval.
In Section~\ref{sec:plu} we provide an efficient algorithm for identifying level-2
strategies with respect to the Plurality rule. Section~\ref{sec:2app}
contains our results for $2$-approval, and in Section~\ref{sec:hard} we present
our hardness results for $k$-approval with $k\ge 4$.  
Section~\ref{sec:conclusions} summarises our results and suggests directions for future work.


\section{Preliminaries}\label{sec:prelim}
In this section we introduce the relevant notation and terminology 
concerning preference aggregation and normal-form games.
\subsection{Preferences and Voting Rules}
We consider $n$-voter elections over a candidate set $C=\{c_1, \dots, c_m\}$; in what follows we use
the terms {\em candidates} and {\em alternatives} interchangeably. Let $\calL(C)$ denote the set
of all linear orders over $C$.
An election is defined by a {\em preference profile} $V=(v_1, \dots, v_n)$, where each $v_i$,
$i\in [n]$, is a linear order over $C$; we refer to $v_i$ as the {\em sincere vote}, or {\em preferences}, of voter $i$.
For two candidates $c_1, c_2\in C$ we write $c_1\succ_i c_2$, if voter~$i$ ranks $c_1$ above $c_2$, and say that voter $i$ {\em prefers} $c_1$ to $c_2$.
For brevity we will sometimes write $ab\dots z$ to represent a vote $v_i$ such that $a\succ_i b\succ_i\cdots\succ_i z$.
We denote the top candidate in $v_i$ by $\tp(v_i)$.  Also, we denote the set of top $k$ candidates in~$v_i$
by $\tp_k(v_i)$; note that $\tp_1(v_i)=\{\tp(v_i)\}$ and
$a\succ_i b$ for all $a\in\tp_k(v_i)$ and $b\in C\setminus\{\tp_k(v_i)\}$.


A (resolute) {\em voting rule} $\calR$ is a mapping that, given a profile $V$, outputs a candidate ${\calR(V)\in C}$,
which we call the {\em winner of the election defined by~$V$}, or simply the {\em winner at~$V$}. 
In this paper we focus on
the family of voting rules known as {\em $k$-approval}.
Under $k$-approval, $k\in [m-1]$, each candidate receives one point from each voter who ranks
her in top $k$ positions; the $k$-approval score of a candidate $c$, denoted by $\scr_k(c, V)$,
is the total number of points that she receives.
The winner is chosen among the candidates with the highest score
according to a fixed tie-breaking  linear order $>$ on the set of candidates $C$: specifically, the winner
is the highest-ranked candidate with respect to this order among the candidates with the highest score.
The $1$-approval voting rule is widely used and known as Plurality. We will denote the $k$-approval rule
(with tie-breaking based on a fixed linear order $>$) by $\calR_k$.
We say that a candidate $x$ {\em beats} candidate $y$ at $V$ with respect to $\calR_k$ and the tie-breaking order $>$
if $\scr_k(x, V) > \scr_k(y, V)$ or $\scr_k(x, V)=\scr_k(y, V)$ and $x>y$.

\subsection{Strategic Voting}\label{sec:strategic}

Given a preference profile $V=(v_1, \dots, v_n)$, and a linear order $v'_i\in\calL(C)$, we denote by $(V_{-i}, v'_i)$
the preference profile obtained from $V$ by replacing $v_i$ with $v'_i$; for readability,
we will sometimes omit the parentheses around $(V_{-i}, v'_i)$ and write $V_{-i}, v'_i$.
We will often use this notation when voter $i$ submits a strategic vote $v'_i$ instead of his sincere vote $v_i$.


\begin{definition}\label{def:GS}
Consider a profile $V=(v_1, \dots, v_n)$, a voter $i$, and a voting rule $\calR$. 
We say that a linear order $v'_i$ is a {\em manipulative vote} of voter $i$ at $V$
with respect to $\calR$ if $\calR(V_{-i}, v'_i)\succ_i \calR(V)$.
We say that $i$ manipulates {\em in favour of candidate $c$} by submitting a vote $v'_i$ if
$c$ is the winner at $\calR(V_{-i}, v'_i)$. 
A voter $i$ is a {\em Gibbard--Sat\-ter\-thwaite manipulator}, or a {\em GS-manipula\-tor}, 
at $V$ with respect to $\calR$ if the set of his manipulative votes at $V$
with respect to $\calR$ is not empty.
We denote the set of all GS-manipulators at $V$ by $N(V, \calR)$.
\end{definition}

Note that a voter may be able to manipulate in favour
of several different candidates. 
Let $F_i=\{c\in C\mid \calR(V_{-i}, v'_i)=c\textrm{ for some }v'_i\in \calL(C)\}$;
we say that the candidates in $F_i$ are {\em feasible for $i$ at $V$ with respect to $\calR$}. 
Note that $F_i\neq\emptyset$ for all $i\in [n]$, as this set contains the $\calR$-winner at $V$
under truthful voting.
We say that two votes $v$ and $v'$ over the same candidate set $C$
are {\em equivalent} with respect to a voting rule $\calR$
if $\calR(V_{-i}, v)=\calR(V_{-i}, v')$ for every voter $i\in[n]$ and every profile $V_{-i}$
of other voters' preferences. 
It is easy to see that
$v$ and $v'$ are equivalent with respect to $k$-approval if and only if $\tp_k(v)=\tp_k(v')$.

\subsection{Normal-form Games}
A {\em normal-form game} $(N, (A_i)_{i\in N}, (\succeq_i)_{i\in N})$  is defined by a set of {\em players} $N$, and,
for each $i\in N$,  a set of {\em strategies} $A_i$ and a preference relation $\succeq_i$
defined on the space of {\em strategy profiles}, i.e., tuples of the form ${\vec s}=(s_1, \dots, s_n)$,
where $s_i\in A_i$ for all $i\in N$\footnote{While one usually defines normal-form games in terms of utility functions,
 defining them in terms of preference relations is more appropriate for our setting,
as preference profiles only provide ordinal information about the voters' preferences.}.
For each pair of strategy profiles ${\vec s}, {\vec t}$ and a player $i\in N$, we write ${\vec s}\succ_i {\vec t}$
if ${\vec s}\succeq_i {\vec t}$ and ${\vec t}\not\succeq_i {\vec s}$.
A normal-form game is viewed as a game of complete and perfect information,
which means that all players are fully aware of the structure of the game they are playing.

Given a strategy profile ${\vec s}=(s_1, \dots, s_n)$ and a strategy $s'_i\in A_i$,
we denote by $({\vec s}_{-i},s'_i)$
the strategy profile $(s_1,\ldots, s'_i,\ldots,s_n)$,
which is obtained from $\vec s$ by replacing $s_i$ with $s'_i$.
We say that a strategy $s_i\in A_i$ {\em weakly dominates} another strategy $s'_i\in A_i$
if  for every strategy profile ${\vec s}_{-i}$ of other players we have 
$({\vec s}_{-i}, s_i)\succeq_i ({\vec s}_{-i},s'_i)$ and
there exists a profile ${\vec s}_{-i}$ of other players' strategies such that
$({\vec s}_{-i}, s_i)\succ_i ({\vec s}_{-i},s'_i)$.


\section{The Model}\label{sec:model}
As suggested in Section~\ref{sec:intro}, our goal is to analyse voting as a strategic game
and consider it from the perspective of the cognitive hierarchy model. As we reason about voters'
strategic behavior, we consider games where players are voters, their strategies 
are ballots they can submit, and their preferences over strategy profiles 
are determined by election outcomes under a given voting rule. We then use the cognitive hierarchy framework
to narrow down the players' strategy sets.

\subsection{Cognitive Hierarchy Framework for Voting Games}
Recall that, in the general framework of cognitive hierarchy, players at level 0 are 
typically assumed to choose their action at random. This is because in general normal-form
games a player who is unable to deliberate about other players' actions 
usually has no reason to prefer one strategy over another.
In contrast, in the context of voting, there is an obvious focal strategy, namely, truthful
voting. Thus, in our model we associate level-0 voters with ideological voters, 
i.e., voters who always vote according to their true preferences.

At the next level of hierarchy are level-1 voters. These voters
assume that all other voters are ideological (i.e., are at level 0), 
and choose their vote so as to get the best outcome they consider possible under this assumption.
That is, voter $i$ votes so as to make his most preferred candidate in $F_i$
the election winner (in particular, if $F_i$ is a singleton, voter $i$ votes truthfully).
We say that a vote $v'_i$ of a voter $i$ 
is a {\em level-1 strategy at profile $V$ with respect to $\calR$} 
if $\calR(V_{-i}, v'_i)\succ_i c$ for all $c\in F_i\setminus\{\calR(V_{-i}, v'_i)\}$.
Note that a level-1 voter that is not a Gibbard--Satterthwaite manipulator 
has no reason to vote non-truthfully, as he does not expect
to be able to change the election outcome according to his tastes;
hence we assume that such voters are truthful.

We are now ready to discuss level-2 voters. These voters believe that all other voters
are at levels 0 or 1 of the cognitive hierarchy. We will further assume that
level-2 voters are agnostic about other voters' levels; thus, from their perspective
every other voter may turn out to be a level-0 voter (which in our setting is equivalent 
to being sincere) or a level-1 voter.
Thus, from the point of view of a level-2 voter, 
a voter who is not a GS-manipulator will stick to his truthful vote, whereas
a GS-manipulator will either choose his action among level-1 strategies or 
(in case he is actually a level-0 voter) vote truthfully.
Thus, when selecting his vote strategically, a level-2 voter takes
into account the possibility that other voters---namely, the GS-manipulators---may be strategic as well.

We further enrich the model by assuming that a level-2 voter may be able
to identify, for each other voter $i$, a subset of level-1 strategies
such that $i$ always chooses his vote from that subset, i.e., a level-2 voter
may be able to rule out some of the level-1 strategies of other voters.
There are several reasons to allow for this possibility. First, the set of all level-1 strategies
for a given voter can be very large, and a voter may be unable or unwilling to identify all 
such votes. 
For example, our level-2 voter may know or believe that other voters use a specific algorithm (e.g., that of \cite{btt89})
to find their level-1 strategies; 
in this case, his set of strategies for each voter $i$ would consist of the truthful vote $v_i$ 
and the output of the respective algorithm. Also, voters may be known not to choose 
manipulations that are (weakly) dominated by other manipulations.  
Finally, voters may prefer not to change their vote beyond what is necessary to make their target 
candidate the election winner, either because they want their vote to be as close to the true preference order as 
possible (see the work of \cite{obr-elk:c:opt}), 
or for fear of unintended consequences of such changes in the complex environment of the game.
Thus, a preference profile together with a voting rule define not just a single game, but a family
of games, which differ in sets of actions available to GS-manipulators.


\subsection{Gibbard--Satterthwaite Games}\label{sec:gs}
We will now describe a formal model that will enable us to reason about the decisions 
faced by a level-2 voter. For convenience, we assume that voter $1$ is a level-2 voter and describe
a normal-form game that captures his perspective of strategic interaction, i.e., his beliefs about the game he is playing.

Fix a voting rule $\calR$, let $V$ be a profile over a set of candidates $C$, 
let $N=N(V,\calR)$ be the set of GS-manipulators at $V$ with respect to $\calR$, and set $N_1=N\cup\{1\}$.
We consider a family of normal-form games defined as follows. 
In each game the set of players is $N_1$, i.e., voter $1$ is a player irrespective
of whether he is actually a GS-manipulator. 
For each player $i\in N_1\setminus\{1\}$, $i$'s strategy set $A_i$ consists of his truthful vote
and a (possibly empty) subset of his level-1 strategies;
for voter $1$ we have $A_1=\calL(C)$, i.e., $1$ can submit an arbitrary ballot.
It remains to describe the voters' preferences over strategy profiles.
For a strategy profile $V^* = (v^*_i)_{i\in N_1}$, where $v^*_i\in A_i$ for $i\in N_1$,
let $V[V^*]=(v'_1, \dots, v'_n)$ be the preference profile
such that $v'_i=v_i$ for $i\not \in N_1$ and $v'_i=v^*_i$ for $i\in N_1$.
Then, given two strategy profiles $V^*$ and $V^{**}$ and a voter $i\in N_1$,
we write $V^*\succeq_i V^{**}$ if and only if  $i$ prefers
$\calR(V[V^*])$ to $\calR(V[V^{**}])$ or $\calR(V[V^*])=\calR(V[V^{**}])$.
We refer to any such game as a {\em GS-game}.  

We denote the set of all GS-games for $V$ and $\calR$ by $\GS(V, \calR)$.
Note that an individual game in $\GS(V, \calR)$ is fully determined by the 
GS-manipulators' sets of strategies, i.e.,
$(A_i)_{i\in N(V, \calR)}$ (player $1$'s set of strategies is always the same, namely, $\calL(C)$). 
Thus, in what follows, we write $G=(V, \calR, (A_i)_{i\in N(V, \calR)})$;
when $V$ and $\calR$ are clear from the context, we simply write $G = (A_i)_{i\in N}$.
We refer to a strategy profile in a GS-game as a {\em GS-profile},
and we will sometimes identify the GS-profile $V^* = (v^*_i)_{i\in N_1}$ with the preference profile
$V[V^*]$. We denote the set of all GS-profiles in a game $G$ by $\GSP(G)$. 

We will now argue that games in $\GS(V, \calR)$ reflect the perspective of voter $1$
when he is at the second level of the cognitive hierarchy. Fix a game $G\in\GS(V, \calR)$.
Note first that, since voter $1$ believes that all other voters belong
to levels 0 and 1 of the cognitive hierarchy, 
he expects all voters who are not GS-manipulators to vote truthfully, i.e., he does
not need to reason about their strategies at all. This justifies having $N_1 = N\cup\{1\}$ as our set of players.
On the other hand, consider a voter $i\in N\setminus\{1\}$. Voter 1 considers it possible
that $i$ is a level-0 voter, who votes truthfully. Voter 1 also entertains
the possibility that $i$ is a level-1 voter, in which case $i$'s vote
has to be a level-1 strategy; as argued above, voter $1$ may also be able to rule
out some of $i$'s level-1 strategies. Consequently, the set $A_i$, which, by definition, contains
$v_i$, consists of all strategies that voter $1$ considers possible for $i$.
Thus, voter $1$'s view of other voters' actions is captured by $G$.

We are now ready to discuss level-2 strategies.
In game-theoretic literature, it is typical to assume that a level-2 player is endowed with probabilistic
beliefs about other players' types as well as a utility function describing his payoffs under all
possible strategy profiles. Under these conditions, it makes sense to define level-2 strategies
as those that maximise player 2's expected payoff. However, in the absence of numerical information, 
as in the case of voting games, we cannot reason about expected payoffs. We can, however, compare
different strategies pointwise, and remove strategies that are weakly dominated by other strategies. 
On the other hand, if a strategy $v$ is not weakly dominated, a level-2 player may hold beliefs
that make him favour $v$, so no such strategy can be removed from consideration
without making additional assumptions about the behavior of players in $N(V, \calR)$.
This reasoning motivates the following definition of a level-2 strategy.

\begin{definition}\label{def:level-2}
Given a GS-game $G=(V, \calR, (A_i)_{i\in N(V, \calR)})$, we say that a strategy $v\in A_1$ of player~1
is a {\em level-2 strategy} if no other strategy of player~1 weakly dominates $v$.
\end{definition}

We note that being weakly undominated is not a very demanding property: a strategy 
can be weakly undominated even if it fares badly in many scenarios. This is illustrated 
by the following example.

\begin{example}\label{ex:lev2}
{\em
Consider the 4-voter profile over $\{a, b, c, d\}$ given in Table~\ref{table:ex-lev2}. 
Suppose that the voting rule
is Plurality and the tie-breaking rule is $a>b>c>d$. As always, we assume that voter 1 
is the level-2 voter. Voters 2, 3, and 4 are GS-manipulators; their most preferred
feasible candidates are, respectively, $d$, $b$, and $c$. Consider the GS-game where
$A_2=\{bdac, dbac\}$, $A_3=\{cbad, bcad\}$, $A_4=\{dcab, cdab\}$. In this game every vote
that does not rank $d$ first is a level-2 strategy for the first player. Indeed, 
a vote that ranks $a$ first is optimal when all other players submit their sincere votes;
a vote that ranks $b$ first is optimal when players 2 and 3 stay sincere, but player 4 votes for $c$;
and a vote that ranks $c$ first is optimal when player 2 votes for $d$, but players 3 and 4 stay sincere.
Note, in particular, that, by changing his vote from $abcd$ (his sincere vote) to $cabd$, 
player $1$ changes the outcome from $a$ (his top choice) 
to $c$ (his third choice) when other players vote truthfully; however, this behavior is rational
if player 1 expects players 3 and 4 (but not player 2) to vote sincerely.

\begin{table}[h]
\begin{center}
\begin{tabular}{c|c|c|c}
voter 1	&voter 2&voter 3&voter 4\\
\midrule
$a$	&$b$	&$c$	&$d$\\
$b$	&$d$	&$b$	&$c$\\
$c$	&$a$	&$a$	&$a$\\
$d$	&$c$	&$d$	&$b$
\end{tabular}
\caption{A profile where voter 1 has three distinct level-2 strategies under Plurality voting.
\label{table:ex-lev2}
}
\end{center}
\end{table}
}
\end{example}

Example~\ref{ex:lev2} illustrates that level-2 strategies are not `safe': there can be circumstances 
where a level-2 strategy results in a worse outcome than sincere voting.
Now, a cautious level-2 player may prefer to stick to his sincere vote
unless he can find a manipulative vote which leads to an outcome that is at least as desirable
as the outcome under truthful voting, for any combination of actions of other players
that our level-2 player considers possible. The following definition, which is motivated
by the concept of safe strategic voting \citep{safe2}, describes the set 
of strategies that even a very cautious level-2 player would prefer to sincere voting.

\begin{definition}
Given a GS-game $G=(V, \calR, (A_i)_{i\in N(V, \calR)})$, we say that a strategy $v\in A_1$ of player~1
is an {\em improving strategy} if $v$ weakly dominates player 1's sincere strategy~$v_1$.
\end{definition}

We note that a level-2 strategy may fail to be an improving strategy, 
and conversely, an improving strategy is not necessarily a level-2 strategy.
For instance, in Example~\ref{ex1} the strategy $bac$
is a level-2 strategy, but not an improving strategy, and none of the level-2 strategies
in Example~\ref{ex:lev2} is improving.
However, it is easy to see that if a player has an improving strategy, 
he also has an improving strategy
that is a level-2 strategy. Moreover, an improving strategy exists if and only
if sincere voting is not a level-2 strategy.

One can also ask if a given strategy weakly dominates all other (non-equivalent) strategies. 
However, while strategies with this property are highly desirable, from the perspective
of a strategic voter it is more important to find out whether his truthful strategy is weakly dominated.
Indeed, the main issue faced
by a strategic voter is whether to manipulate at all, and if a certain vote
can always ensure an outcome that is at least as good, and sometimes better, 
as that guaranteed by his truthful vote, 
this is a very strong incentive to use it, even if another non-truthful vote may be better
in some situations. 
This issue is illustrated by Example~\ref{ex:lowerranked} below, which 
describes a profile where a player has two incomparable improving strategies.

\begin{example}\label{ex:lowerranked}
{\em
Let the profile of sincere preferences be as in Table~\ref{table:example}, and assume that 
the voting rule is Plurality and the tie-breaking order is given by $w>d>c>b>a$. 
The winner at the sincere profile is $w$. All level-1 strategies of voter 2 are equivalent to
$cbdwa$, whereas all level-1 strategies of voter 3 are equivalent to
$dcbwa$; voters 4 and 5 are not GS-manipulators.
Consider the GS-game where for $i\in\{2, 3\}$
the set of strategies of player $i$ consists of his truthful vote and all of his level-1 strategies.
Voter 1, who is our level-2 player, 
can manipulate either in favour of $b$ or in favour of $d$, by ranking the respective candidate first. 
Indeed, for player 1 both $badwc$ and $dabwc$ weakly dominate truthtelling.
However, neither of these strategies weakly dominates the other:
$badwc$ is preferable if no other player uses a level-1 strategy, whereas $dabcw$
is preferable if player 2 uses his level-1 strategy, but player 3 votes sincerely.

\begin{table}[h]
\begin{center}
\begin{tabular}{c|c|c|c|c}
voter 1&voter 2&voter 3&voter 4&voter 5\\
\midrule
$a$&$b$&$c$&$d$&$w$\\
$b$&$c$&$d$&$w$&$a$\\
$d$&$d$&$b$&$c$&$b$\\
$w$&$w$&$w$&$a$&$c$\\
$c$&$a$&$a$&$b$&$d$\\
\end{tabular}
\caption{Player 1 has two incomparable improving strategies.
}\label{table:example}
\end{center}
\end{table}
}
\end{example}

We note that a level-2 voter may find it useful to act as a counter-manipulator
\citep{pattanaik1976threats,GHRS}, i.e., 
to submit a vote that is not manipulative with respect to the truthful
profile, but minimises the damage from someone else's manipulation.

\begin{example}\label{ex:counter}
{\em
Let the profile of sincere preferences be as in Table~\ref{table:example2}, and assume that 
the voting rule is Plurality and the tie-breaking order is given by $a>b>c$. 
Under truthful voting $a$ wins, so voter 6 is the only GS-manipulator: if he changes his vote
to $bca$ then $b$ wins, and $b\succ_6 a$. Therefore, for voter 1 voting $acb$ is preferable to voting truthfully:
this insincere vote has no impact if voter 6 votes truthfully, but prevents $b$ from becoming a winner when voter 6
submits a manipulative vote. 

Thus, in this example $acb$ is an improving strategy, 
and truthful voting is not a level-2 strategy as it is weakly dominated
by voting $acb$. In contrast, $acb$ is a level-2 strategy, as no other strategy weakly dominates it.

\begin{table}[h]
\begin{center}
\begin{tabular}{c|c|c|c|c|c}
voter 1&voter 2&voter 3&voter 4 &voter 5 &voter 6\\
\midrule
$c$&$a$&$a$&$b$&$b$&$c$\\
$a$&$b$&$b$&$a$&$a$&$b$\\
$b$&$c$&$c$&$c$&$c$&$a$\\
\end{tabular}
\caption{Countermanipulation under Plurality.}\label{table:example2}
\end{center}
\end{table}
}
\end{example}


\subsection{Algorithmic Questions}
From an algorithmic perspective, perhaps the most natural questions 
suggested by our framework are how to decide whether a given strategy
is a level-$2$ strategy, or how to compute a level-$2$ strategy.
A related question is whether a given strategy is an improving strategy
and whether an improving strategy can be efficiently computed.
These questions offer an interesting challenge from an algorithmic perspective:
the straightforward algorithm for deciding whether a given strategy is a level-2 strategy
or an improving strategy relies on considering all combinations of other players' strategies,
and hence has exponential running time. It is therefore natural to ask 
whether for some voting rules exhaustive choice can be avoided.
We explore this question in Sections~\ref{sec:plu}--\ref{sec:hard};
for concreteness, we focus on $k$-approval, for various values of $k$.


\section{Level-1 Strategies Under $\boldsymbol{k}$-Approval}\label{sec:classification}
The goal of this section is to understand and classify level-1 strategies under the $k$-approval voting rule;
this will help us reason about level-2 strategies in subsequent sections.
In what follows, we fix a linear order $>$ used for tie-breaking. We start with a simple, but useful lemma.

\begin{lemma}
\label{mantypes}
Fix $k\ge 1$.
Consider a profile $V$ over $C$, let $w$ be the $k$-approval winner at $V$, 
and let $x$ be an alternative in $C\setminus\{w\}$.
Then any manipulative vote by a voter $i$ in favour of $x$ at $V$
falls under one of the following two categories:
\begin{description}
\item[Type 1]
Voter $i$ increases the score of $x$ by 1 without decreasing the score of  $w$.
In this case $w, x\not\in\tp_k(v_i)$, $x\succ_i w$,
and the manipulative vote $v'_i$ satisfies $x\in\tp_k(v'_i)$, $w\not\in\tp_k(v'_i)$.
In such cases voter $i$ will be referred to as a {\em promoter} of~$x$.
\item[Type 2] 
Voter $i$ decreases the score of $w$ (and possibly that of some other alternatives) by 1
without increasing the score of $x$.
In this case $w, x\in\tp_k(v_i)$, $x\succ_i w$,
and the manipulative vote $v'_i$ satisfies $x\in\tp_k(v'_i)$, $w\not\in\tp_k(v'_i)$.
In such cases voter $i$ will be referred to as a {\em demoter} of~$w$.
Manipulations of type 2 only exist for $k\ge 2$.
\end{description}
\end{lemma}

\begin{proof}
Suppose that voter $i$ manipulates in favour of $x$.
If $i$ can increase the score of $x$, then $x\not\in\tp_k(v_i)$.
However, $i$ must rank $x$ higher than $w$ (otherwise, this would not be a manipulation).
Thus, $w\not\in\tp_k(v_i)$ and therefore voter $i$ cannot decrease $w$'s score.
Moreover, if $w\in\tp_k(v'_i)$, then $w$ would beat $x$ under $k$-approval in $(V_{-i}, v'_i)$;
thus, $w\not\in\tp_k(v'_i)$. 

On the other hand, suppose that $i$ cannot increase the score of $x$.
This means that $x\in\tp_k(v_i)$ and hence $i$
is left with reducing the scores of some of $x$'s competitors including the current winner $w$.
For this to be possible, it has to be the case that $w\in\tp_k(v_i)$ and $x\succ_i w$.
Also, we have $w\not\in\tp_k(v'_i)$, as otherwise $w$ would beat $x$ under $k$-approval in $(V_{-i}, v'_i)$.
Finally, as $w\neq x$, we can only have $x, w\in\tp_k(v_i)$ if $k\ge 2$.
\end{proof}

The classification in Lemma~\ref{mantypes} justifies our terminology:
a promoter promotes a new winner and a demoter demotes the old one.
Under Plurality, i.e., when $k=1$, we only have promoters. 

Let $X = \{x_1, \dots, x_\ell\}$ and $Y = \{y_1, \dots, y_\ell\}$ be two disjoint sets of candidates.
Given a linear order $v$ over $C$, we denote by $v[X; Y]$ the vote obtained
by swapping $x_j$ with $y_j$ for $j\in[\ell]$.
If the sets $X$ and $Y$ are singletons, i.e., $X=\{x\}$, $Y=\{y\}$,
we omit the curly braces, and simply write $v[x;y]$.
Clearly, under $k$-approval any manipulative vote of voter $i$ is equivalent 
to a vote of the form $v_i[X; Y]$, where $X\subseteq\tp_k(v_i)$, $Y\subseteq C\setminus\tp_k(v_i)$. 
We can now state a corollary of Lemma~\ref{mantypes}, which
characterises the possible effects of a manipulative vote under 
$k$-approval.

\begin{corollary}\label{prop:change}
Let $w$ be the $k$-approval winner at a profile $V$, let $v^*_i=v_i[X; Y]$, 
where $X\subseteq \tp_k(v_i)$ and $Y\subseteq C\setminus\tp_k(v_i)$.
Let $V'=(V_{-i},v^*_i)$, and let $w'\ne w$ be the $k$-approval winner at $V'$. 
Then either  $w\in X$ or $w'\in Y$ but not both.  
\end{corollary}

\begin{proof}
Lemma~\ref{mantypes} implies that either the new winner is promoted or the old winner is demoted, but not both.
\end{proof}

Consider a manipulative vote $v_i[X;Y]$ of voter $i$ at $V$ under $k$-approval;
we say that $v_i[X;Y]$ is {\em minimal} if for every other manipulative vote $v_i'$
of voter $i$ there is a vote $v_i[X';Y']$ that is equivalent to $v'_i$ and satisfies $|X'|\ge |X|$. That is, a 
manipulative vote is minimal if it performs as few swaps as possible. 
Arguably, minimal manipulative votes are the main tool that a rational voter would use,
as they achieve the desired result in the most straightforward way possible.

We now introduce some useful notation.
Fix a profile $V$. Let $w$ be the $k$-approval winner at $V$, 
and let $t = \scr_k(w, V)$. Set
\begin{eqnarray*}
S_1(V, k) &=& \{c\in C\mid \scr_k(c)=t, w>c\},\\ 
S_2(V, k) &=& \{c\in C\mid \scr_k(c)=t-1, c>w\},
\end{eqnarray*}
and set $S(V, k)=S_1(V, k)\cup S_2(V, k)$. 

We say that a candidate $c$ is {\em $k$-competitive} at $V$ if $c\in S(V, k)$.
The following proposition explains our choice of the term: only $k$-competitive
candidates can become $k$-approval winners as a result of a manipulation.

\begin{proposition}
\label{pinSVk}
Suppose that some voter can manipulate in favour of a candidate $p\in C$
at a profile $V$ with respect to $k$-approval. 
Then $p\in S(V, k)$.
\end{proposition}
\begin{proof}
Let $w$ be the $k$-approval winner at $V$; clearly, $w\neq p$.
Suppose that voter $i$ can manipulate in favour of $p$ at $V$ by submitting a vote $v'_i$;
let $V'=(V_{-i}, v'_i)$.
Set $t=\scr_k(w, V)$; then $\scr_k(p, V)\le t$.
Note that if $\scr_k(p, V)=t$, it has to be the case that $w > p$, since otherwise
$p$ would beat $w$ at $V$. Thus, in this case $p\in S_1(V, k)$.
Now, suppose that $\scr_k(p, V)=t-1$. 
By Corollary~\ref{prop:change} we have either $\scr_k(w, V')=\scr_k(p, V')=t$ (if $p$ was promoted)
or $\scr_k(w, V')=\scr_k(p, V')=t-1$ (if $w$ was demoted). In both cases we have to have $p > w$, 
as otherwise $w$ would beat $p$ at $V'$. Therefore, in this case $p\in S_2(V, k)$.
Finally, note that it cannot be the case that  $\scr_k(p, V)\le t-2$, 
since in this case by Corollary~\ref{prop:change}
we would have either $\scr_k(w, V')\ge t-1$, $\scr_k(p, V')\le t-2$
or $\scr_k(w, V')=t$, $\scr_k(p, V')\le t-1$, i.e., $w$ would beat $p$ at $V'$.
\end{proof}


Suppose that $S(V, k)\neq\emptyset$.
If $S_1(V,k)\ne \emptyset$, then by $p^*(V, k)$ we denote the top-ranked candidate in $S_1(V,k)$ with respect to $>$;
otherwise, we denote by $p^*(V, k)$ the top-ranked candidate in $S_2(V,k)$ with respect to $>$.
Thus, $p^*(V, k)$ beats all candidates other than $w$ at $V$, and would become a winner
if it were to gain one point or if $w$ were to lose one point.
We omit $V$ and $k$ from the notation when they are clear from the context.

We are now ready to embark on the computational complexity analysis of level-2 strategies 
under $k$-approval, for various values of $k$.

\section{Plurality}\label{sec:plu}
Plurality voting rule is $k$-approval with $k=1$. 
For this rule we only have manipulators of type 1, and
all manipulative votes of voter $i$ in favour of candidate $c$
are equivalent: in all such votes $c$ is placed in the top position. 

The main result of this section is that the problem of deciding whether a given strategy 
of voter 1 weakly dominates another strategy of that voter admits a polynomial-time algorithm. Note that, 
since under Plurality there are only $m$ votes that are pairwise non-equivalent, 
this means that we can check if a given strategy is a level-2 strategy or an improving strategy, or
find a level-2 strategy or an improving strategy (if it exists) in polynomial time; we formalise this intuition in 
Corollary~\ref{cor:lev2} at the end of this section.

Fix a preference profile $V$ over a candidate set $C$ and consider a GS-game 
$(V, \calR_1, (A_i)_{i\in N})$, where $N=N(V, \calR_1)$. Let $w$ be the Plurality winner at $V$. 
As argued above, for each $i\in N\setminus\{1\}$ the set $A_i$ consists of $v_i$ and possibly a number of pairwise
equivalent manipulative votes; without loss of generality, we can remove
all but one manipulative vote, so that $|A_i|\le 2$ for all $i\in N\setminus\{1\}$.
We will now explain how, given two votes $v_1'$ and $v_1''$, voter $1$ 
can efficiently decide if one of these votes weakly dominates the other.

We will first describe a subroutine that will be used by our algorithm.
\begin{lemma}\label{lem:plur-flow}
There is a polynomial-time procedure 
$$
\calA=\calA(G, r, r', x, y, C^{[1]}, C^{[0]}, C^{[-1]}, C^{[-2]})
$$ 
that, 
given a GS-game $G=(V, \calR_1, (A_i)_{i\in N(V, \calR_1)})$ with $|V|=n$, 
two integers $r, r'\in\{0, \dots, n\}$, two distinct candidates $x, y\in C$,
and a partition of candidates in $C\setminus\{x, y\}$ into $C^{[1]}$, $C^{[0]}$, $C^{[-1]}$ and $C^{[-2]}$,  
decides whether there is a strategy profile $V^*$ in $G$ such that 
\begin{itemize}
\item
$\scr_1(x, V[V^*]_{-1})=r$,
\item
$\scr_1(y, V[V^*]_{-1})=r'$, and
\item
for each $c\in C\setminus\{x, y\}$ and each $\ell\in\{1, 0, -1, -2\}$
if $c\in C^{[\ell]}$ then 
$\scr_1(c, V[V^*]_{-1})\le r+\ell$.
\end{itemize}
\end{lemma}
\begin{proof}
We proceed by reducing our problem to an instance of network flow with capacities and lower bounds, as follows.
We construct a source, a sink, a node for each voter $i\in[n]\setminus\{1\}$ and a node 
for each candidate in $C$. There is an arc from the source to each voter node; the capacity and the lower
bound of this arc are set to $1$, i.e., it is required to carry one unit of flow.
Also, there is an arc with capacity $1$ and lower bound 0 from voter $i$ to candidate $c$ if 
$i\in N(V, \calR_1)\setminus\{1\}$ and $c=\tp(v)$ for some $v\in A_i$ or if $i\in [n]\setminus(N(V, \calR_1)\cup\{1\})$ 
and $c=\tp(v_i)$. Finally, there is an arc from each candidate $c$ to the sink. The capacity of this arc
is set to $r+\ell$ if $c\in C^{[\ell]}$ for some $\ell\in\{1, 0, -1, -2\}$;
the lower bounds for these arcs are $0$. For $x$, 
both the capacity and the lower bound of the arc to the sink are set to $r$, and for $y$
they are both set to $r'$. 
We note that some of the capacities may be negative, in which case there is no valid flow.
It is immediate that an integer flow that satisfies all constraints
corresponds to a strategy profile in $G$ where all candidates have the required scores;
it remains to observe that the existence of a valid integer flow  
can be decided in polynomial time.
\end{proof}

We are now ready to describe our algorithm.

\begin{theorem}\label{thm:plur}
Given a GS-game $G=(V, \calR_1, (A_i)_{i\in N(V, \calR_1)})$
and two strategies $v_1', v_1''\in\calL(C)$ of player 1 we can decide in polynomial time
whether $v_1'$ weakly dominates $v_1''$.
\end{theorem}
\begin{proof}
We will design a polynomial-time procedure that, given two strategies $u, v$ of player~1,
decides if there exists a profile $V^*_{-1}$
of other players' strategies such that $\calR_1(V[V^*_{-1}, u])\succ_1 \calR_1(V[V^*_{-1}, v])$;
by definition, $v'_1$ weakly dominates $v''_1$ if this procedure returns `yes' for $u=v_1'$, $v=v_1''$
and `no' for $u=v_1'', v=v_1'$.
 
Let $a=\tp(u)$, $b=\tp(v)$. We can assume without loss of generality that $a\neq b$, 
since otherwise $u$ and $v$ are equivalent with respect to Plurality. 
Consider an arbitrary profile $V^*_{-1}$ of other players'
strategies, and let $V^u=V[V^*_{-1}, u]$, $V^v=V[V^*_{-1}, v]$, $w^u=\calR_1(V^u)$, $w^v=\calR_1(V^v)$. 
We note that $w^u\neq a$ implies $w^v\neq a$:
if $w^u$ beats $a$ at $V^u$, this is also the case at $V^v$. Similarly, 
if $w^v\neq b$ then also $w^u\neq b$. Now, suppose that 
$w^u\neq a$ and $w^v\neq b$. We claim that in this case $w^u=w^v$.
Indeed, suppose for the sake of contradiction that
$w^u\neq w^v$. As $w^u\neq a$, $w^v\neq b$, the argument above shows that $\{w^u, w^v\}\cap \{a, b\}=\emptyset$.
Thus, both $w^u$ and $w^v$ have the same Plurality score at $V^u$ and $V^v$; as $w^u$ beats $w^v$
at $V^u$, this must also be the case at $V^v$, a contradiction.

Note that $\calR_1(V[V^*_{-1}, u])\succ_1 \calR_1(V[V^*_{-1}, v])$ 
if and only if $w^u\succ_1 w^v$. By the argument in the previous paragraph, 
this can happen in one of the following three cases:
(i)   $w^u=a$, $w^v=b$ and $a\succ_1 b$;
(ii)  $w^u=a$, $w^v=w$ for some $w\neq b$, $a\succ_1 w$;
(iii) $w^u=w$, $w^v=b$ for some $w\neq a$, $w\succ_1 b$.
(We note that we can merge case (i) into case (ii) or case (iii);
we choose not to do so for the sake of clarity of presentation.)
We will now explain how to check if there exists a profile $V^*_{-1}$
that corresponds to any of these three situations.

\begin{description}
\item
Case (i): $w^u=a$, $w^v=b$. 

Suppose first that $a > b$. Then a desired profile $V^*_{-1}$ exists if and only if there is some value $t\in[n]$
such that $\scr_1(a, V^u)=t$ and 
\begin{enumerate}
\item[(a)] 
$\scr_1(b, V^u)=t$, $\scr_1(c, V^u)\le t$ for all $c\in C\setminus\{a, b\}$ 
with $a>c$, $\scr(c, V^u)\le t-1$ for all $c\in C\setminus\{a, b\}$ with $c>a$, or
\item[(b)] 
$\scr_1(b, V^u)=t-1$, $\scr_1(c, V^u)\le t$ for all $c\in C\setminus\{a, b\}$
with $b>c$, and $\scr(c, V^u)\le t-1$ for all $c\in C\setminus\{a, b\}$ with $c>b$.
\end{enumerate}
Note that $\scr_1(a, V^u_{-1})=\scr_1(a, V^u) - 1$ and $\scr_1(c, V^u_{-1})=\scr_1(c, V^u)$
for $c\in C\setminus\{a\}$.
Thus, to check if condition (a) is satisfied for some $t\in [n]$, 
we set $C^{[1]} =\{c\in C\setminus\{a, b\}\mid a > c\}$, $C^{[0]} =\{c\in C\setminus\{a, b\}\mid c > a\}$,
$C^{[-1]}=C^{[-2]}=\emptyset$
and call 
$$
\calA(G, t-1, t, a, b, C^{[1]}, C^{[0]}, C^{[-1]}, C^{[-2]}).
$$
Similarly, to determine whether condition (b) is satisfied for some $t\in [n]$, 
we set $C^{[1]} =\{c\in C\setminus\{a, b\}\mid b>c\}$, $C^{[0]} =\{c\in C\setminus\{a, b\}\mid c > b\}$,
$C^{[-1]}=C^{[-2]}=\emptyset$
and call
$$
\calA(G, t-1, t-1, a, b, C^{[1]}, C^{[0]}, C^{[-1]}, C^{[-2]}).
$$
The answer is `yes' if one of these calls returns `yes' for some $t\in [n]$.

For the case $b>a$ the analysis is similar. In this case, we need to decide whether there exists a value of 
$t\in [n]$ such that
$\scr_1(a, V^u)=t$ and 
\begin{itemize}
\item[(a)] 
$\scr_1(b, V^u)=t-1$, $\scr_1(c, V^u)\le t$ for all $c\in C\setminus\{a, b\}$
with $a>c$, and $\scr(c, V^u)\le t-1$ for all $c\in C\setminus\{a, b\}$ with $c>a$, or
\item[(b)] $\scr_1(b, V^u)=t-2$, $\scr_1(c, V^u)\le t-1$ for all $c\in C\setminus\{a, b\}$
with $b>c$, and $\scr(c, V^u)\le t-2$ for all $c\in C\setminus\{a, b\}$ with $c>b$.
\end{itemize}
Again, this can be decided by calling the procedure $\calA$ with appropriate parameters;
we omit the details.

\item
Case (ii): $w^u=a$, $w^v=w$ for some $w$ with $a\succ_1 w$.
In this case, we go over all candidates $w\in C\setminus\{a, b\}$ with $a\succ_1 w$
and all values of $t\in [n]$ and call $\calA$ with appropriate parameters.

Specifically, if $a>w$, we start by setting $r =t-1, r'=t$, and
$$  
C^{[1]} =  \{c\in C\setminus\{a, w, b\}\mid w>c\}, 
C^{[0]} = \{c\in C\setminus\{a, w, b\}\mid c>w\}, 
C^{[-1]} = C^{[-2]} = \emptyset.
$$
We then place $b$ in $C^{[0]}$ if $w>b$ and in $C^{[-1]}$ otherwise;
our treatment of $b$ reflects the fact that she gets an extra point at $V^v$.

If $w>a$ we start by setting $r=t-1, r'=t-1$, and
$$
C^{[1]} =  \emptyset,  
C^{[0]} = \{c\in C\setminus\{a, w, b\}\mid w>c\},
C^{[-1]} = \{c\in C\setminus\{a, w, b\}\mid c>w\},
C^{[-2]} = \emptyset.
$$
We then place $b$ in $C^{[-1]}$ if $w>b$ and in $C^{[-2]}$ otherwise.

Finally, we call
$$
\calA(G, r, r', a, w, C^{[1]}, C^{[0]}, C^{[-1]}, C^{[-2]}).
$$
The answer is `yes' if one of these calls returns `yes' for some $t\in [n]$
and some $w$ with $a\succ_1 w$.

\item
Case (iii): $w^u=w$, $w^v=b$ for some $w$ with $w\succ_1 b$.
The analysis is similar to the previous case; we omit the details.\qedhere
\end{description}
\end{proof}

Theorem~\ref{thm:plur} immediately implies that natural questions concerning level-2 strategies 
and improving strategies are computationally easy.

\begin{corollary}\label{cor:lev2}
Given a GS-game $G=(V, \calR_1, (A_i)_{i\in N(V, \calR_1)})$
and a strategy $v_1'\in \calL(C)$ of player~1 we can decide in polynomial time
whether $v_1'$ is a level-2 strategy or an improving strategy. Moreover, we can decide in polynomial
time whether player~1 has a level-2 strategy or an improving strategy in $G$.
\end{corollary}
\begin{proof}
Let $a=\tp(v'_1)$.
To decide whether $v'_1$ is an improving strategy, we use the algorithm 
described in the proof of Theorem~\ref{thm:plur} to check whether $v'_1$ weakly dominates $v_1$.
Similarly, to decide whether $v'_1$ is a level-2 strategy, for each $c\in C\setminus\{a\}$
we construct a vote $v^c$ with $\tp(v^c)=c$ and check whether $v^c$
weakly dominates $v'_1$ using the algorithm from the proof of Theorem~\ref{thm:plur}. 
As every strategy of player $1$ is equivalent either to $v'_1$ or to one of the votes we constructed, 
$v'_1$ is a level-2 strategy if and only if it is not weakly dominated by any of the votes
$v^c$, $c\in C\setminus\{a\}$.

Similarly, to decide whether $v_1$ has a level-2 strategy (respectively, an improving strategy), 
we consider all of his $m$ pairwise non-equivalent strategies, and check if any of them is a level-2 strategy
(respectively, an improving strategy), as described in the previous paragraph. 
\end{proof}

\section{2-Approval}\label{sec:2app}
In this section, we study the computational complexity of identifying level-2 strategies 
and improving strategies in GS-games under $2$-approval. We show that if the level-2 player believes
that level-1 players can only contemplate minimal manipulations, 
he can efficiently compute his level-2 strategies as well as his improving strategies. 
As argued in Section~\ref{sec:classification}, 
minimality is a reasonable assumption, as level-1 players have no reason to use complex
strategies when simple strategies can do the job.

Specifically, we prove that,  
under the minimality assumption, given two strategies $v'$ and $v''$, 
the level-2 player can decide in polynomial time whether one of these strategies weakly dominates the other;
just as in the case of Plurality, this implies that he 
can check in polynomial time whether
a given strategy is a level-2 (respectively, improving) strategy or identify all of his level-2 
(respectively, improving) strategies.

The following observations play a crucial role in our analysis.

\begin{proposition}\label{prop:2app-demote}
Consider a GS-game $G=(V, \calR_2, (A_i)_{i\in N(V, \calR_2)})$.
Let $w$ be the 2-approval winner at $V$. 
Then for each player $i\in N(V, \calR_2)\setminus\{1\}$
such that $w\in\tp_2(v_i)$ it holds that $\tp(v_i)\neq w$ and 
the candidate $\tp(v_i)$ is ranked in top two positions in every vote $v\in A_i$.
\end{proposition}

Proposition~\ref{prop:2app-demote} 
concerns voters who are demoters, and follows immediately
from Lemma~\ref{mantypes}; note also that it does not depend on the minimality assumption.

\begin{proposition}\label{prop:2app-promote}
Consider a GS-game $G=(V, \calR_2, (A_i)_{i\in N(V, \calR_2)})$.
Let $w$ be the 2-approval winner at $V$. 
Consider a player $i\in N(V, \calR_2)\setminus\{1\}$ 
such that $w\not\in\tp_2(v_i)$ and the set $A_i$ consists of $i$'s truthful vote 
and a subset of $i$'s minimal manipulations.
Let $\tp_2(v) = \{a, a'\}$.
Then there is a candidate $c\in C\setminus\{a, a'\}$
such that for each $v\in A_i$ we have $\tp_2(v)\in\{\{a, a'\}, \{a, c\}, \{a',c\}\}$. 
\end{proposition}

\begin{proof}
Player $i$ cannot lower the score of $w$ by changing his vote. However, he can raise the scores
of some candidates in $C\setminus\tp_2(v_i)$ by moving these candidates into top two positions.
In general, $i$ can do that for two candidates simultaneously; however, the minimality
assumption implies that $i$ only moves one candidate into the top two positions. 
Thus, $i$ is a promoter (see Section~\ref{sec:prelim}).
For a vote $v'_1$ to be a level-1 strategy 
the promoted candidate has to be $i$'s most preferred candidate in $S(V,2)\setminus \tp_2(v_i)$ 
(let us denote this candidate by $p$). 
Thus, in this case voter $i$ has 3 options:
(1) to vote truthfully, 
(2) to swap $p$ with the candidate that he ranks first or 
(3) to swap $p$ with the the candidate he ranks second.
This completes the proof
\end{proof}

Propositions~\ref{prop:2app-demote} and~\ref{prop:2app-promote}
enable us to establish an analogue of Lemma~\ref{lem:plur-flow}
for 2-approval under the minimality assumption.

\begin{lemma}\label{lem:2app-flow}
There is a polynomial-time procedure 
$$
\calA'=\calA'(G, r, r', x, y, C^{[1]}, C^{[0]}, C^{[-1]}, C^{[-2]})
$$ 
that, 
given a GS-game $G=(V, \calR_2, (A_i)_{i\in N(V, \calR_2)})$ with $|V|=n$,
where for each $i\in N\setminus\{1\}$ the set $A_i$ consists of $i$'s truthful vote 
and a subset of $i$'s minimal manipulations, 
two integers $r, r'\in\{0, \dots, n\}$, two distinct candidates $x, y\in C$,
and a partition of candidates in $C\setminus\{x, y\}$ into $C^{[1]}$, $C^{[0]}$, $C^{[-1]}$ and $C^{[-2]}$,  
decides whether there is a strategy profile $V^*$ in $G$ such that 
\begin{itemize}
\item
$\scr_2(x, V[V^*]_{-1})=r$,
\item
$\scr_2(y, V[V^*]_{-1})=r'$, and
\item
for each $\ell\in\{1, 0, -1, -2\}$
and each $c\in C^{[\ell]}$ it holds that 
$\scr_2(c, V[V^*]_{-1})\le r+\ell$.
\end{itemize}
\end{lemma}
\begin{proof}
Let $w$ be the 2-approval winner at $V$. If $S(V, 2)\neq\emptyset$, set $p^*=p^*(V, 2)$.
We use essentially the same construction as in the proof of Lemma~\ref{lem:plur-flow}.
Specifically, the set of nodes consists of a sink, a source, one node for each
voter in $[n]\setminus\{1\}$, and one node for each candidate $c\in C$.
For each $i\in N\setminus\{1\}$, the capacity and the lower bound 
of the arc from the source to node $i$ are equal to $2$, and the capacities 
and lower bounds of the arcs from candidates 
to the source are defined as in the proof of Lemma~\ref{lem:plur-flow}. 
It remains to describe the arcs connecting voters and candidates.

If $i\not\in N$, we add an arc from $i$ to $c$ for each $c\in\tp_2(v_i)$;
the capacity and the lower bound of these arcs are 1, encoding the fact that
$i$ has to vote for his top 2 candidates.

Now, consider a voter $i\in N\setminus\{1\}$ who is a demoter;
if such a voter exists, we have $S(V, 2)\neq\emptyset$ and hence $p^*$ is well-defined. 
By Proposition~\ref{prop:2app-demote} we have $\tp_2(v_i)=\{p^*, w\}$
and $A_i=\{v_i[w;c]\mid c\in C_i\}$ for some $C_i\subset C\setminus\{p^*, w\}$. 
Then we introduce an arc from $i$ to $p^*$ whose capacity and lower bound are both set to 1, 
and arcs with capacity 1 and lower bound 0 from $i$ to each $c\in C_i\cup\{w\}$.

Finally, consider a voter $i\in N\setminus\{1\}$ who is a promoter;
let $\tp_2(v_i)=\{a, a'\}$ and let $p$ be $i$'s most preferred candidate
in $S(2, V)\setminus\{a, a'\}$. If $A_i$ contains both a vote $v'$
with $\tp_2(v')=\{a, p\}$ and a vote $v''$ with $\tp_2(v'')=\{a', p\}$, then
by Proposition~\ref{prop:2app-promote} it suffices to add arcs with capacity 1 and lower bound
0 that go from $i$ to $a$, $a'$, and $p$. If we have $\tp_2(v)\in\{\{a, a'\}, \{a, p\}\}$
for each $v\in A_i$, we add an arc with capacity 1 and lower bound 1 from $i$ to $a$
and arcs with capacity 1 and lower bound 0 from $i$ to $a'$ and $p$.
Similarly, if we have $\tp_2(v)\in\{\{a, a'\}, \{a', p\}\}$
for each $v\in A_i$, we add an arc with capacity 1 and lower bound 1 from $i$ to $a'$ 
and arcs with capacity 1 and lower bound 0 from $i$ to $a$ and $p$.

It is clear from the construction that a valid integer flow in this network 
corresponds to a strategy profile $V^*$ with the desired properties.
\end{proof}

We are now ready to prove the main result of this section.

\begin{theorem}\label{thm:2-app}
Given a GS-game $G=(V, \calR_2, (A_i)_{i\in N(V, \calR_2)})$,
where for each $i\in N\setminus\{1\}$ the set $A_i$ consists of $i$'s truthful vote
and a subset of $i$'s minimal manipulations,
and two strategies $v_1', v_1''\in\calL(C)$ of player~1, 
we can decide in polynomial time whether $v_1'$ weakly dominates $v_1''$.
\end{theorem}

\begin{proof}
Just as in the proof of Theorem~\ref{thm:plur}, it suffices to design a polynomial-time procedure that, 
given two strategies $u, v$ of player~1, decides if there exists a profile $V^*_{-1}$
of other players' strategies such that $\calR_2(V[V^*_{-1}, u])\succ_1 \calR_2(V[V^*_{-1}, v])$.
Let $\tp_2(u)=\{a, a'\}$, $\tp_2(v) = \{b, b'\}$. We can assume that $\{a, a'\}\neq \{b, b'\}$, 
and we will focus on the case where $\{a, a'\}\cap \{b, b'\}=\emptyset$; the case 
where $\{a, a'\}\cap \{b, b'\}$ is a singleton is similar (and simpler).

We use the same notation as in the proof of Theorem~\ref{thm:plur}:
given a profile $V^*_{-1}$ of other players' strategies, 
we let $V^u=V[V^*_{-1}, u]$, $V^v=V[V^*_{-1}, v]$, $w^u=\calR_2(V^u)$, $w^v=\calR_2(V^v)$.
Our goal then is to decide if there exists a profile $V^*_{-1}$ such that $w^u\succ_1 w^v$.
To this end, we go over all values of $t\in[n-1]$ and all candidates $w, w'\in C$
with $w\succ_1 w'$, and ask if there is a profile $V^*_{-1}$
such that $w$ wins at $V[V^*_{-1}, u]$ with $t$ points, 
whereas $w'$ wins at $V[V^*_{-1}, v]$. As in the proof of Theorem~\ref{thm:plur}, 
for each triple $(t, w, w')$ 
we have to consider a number of possibilities, depending on whether $w\in\{a, a'\}$, 
$w'\in\{b, b'\}$ as well as on the relative positions of $w$, $w'$, $a$, $a'$, $b$, and $b'$
with respect to the tie-breaking order. The analysis is as tedious as it is straightforward;
to illustrate the main points, we consider two representative cases. 

\begin{description}
\item[$\boldsymbol{w=a, w'=b, a>a'>b>b'}$\ \ ] 
In this case,  $a$ wins with $t$ points at $V^u$ 
if and only if $\scr_2(a', V^u_{-1})\le t-1$ and for each $c\in C\setminus\{a, a'\}$ we have
$\scr_2(c, V^u_{-1})\le t$ if $a>c$ and $\scr_2(c, V^u_{-1})\le t-1$ if $c>a$.
Suppose that these conditions are satisfied.
Then $b$ can win at $V^v$ with $t+1$ or $t$ points. The former case is possible
if and only if $\scr_2(b, V^u_{-1})= t$. 
The latter case is possible
if and only if $\scr_2(b, V^u_{-1})= t-1$, $\scr_2(b', V^u_{-1})\le t-1$, 
and for each $c\in C\setminus\{a, b, b'\}$ we have
$\scr_2(c, V^u_{-1})\le t$ if $b>c$ and $\scr_2(c, V^u_{-1})\le t-1$ if $c>b$.

Thus, to decide whether this situation is possible, we have to call $\calA'$ twice.
For our first call, we set 
$C^{[1]}=\{c\in C\setminus\{a', b\}\mid a > c\}$, 
$C^{[0]}=\{c\in C\setminus\{a'\}\mid c > a\}\cup \{a'\}$, 
$C^{[-1]}=C^{[-2]}=\emptyset$ and call
$$
\calA'(G, a, b, t-1, t, C^{[1]}, C^{[0]}, C^{[-1]}, C^{[-2]}).
$$
For our second call, we set $C^{[1]}=\{c\in C\setminus\{b'\}\mid b > c\}$, 
$C^{[0]}= \{c\in C\setminus\{a\}\mid c > b\}\cup\{b'\}$, 
$C^{[-1]}=C^{[-2]}=\emptyset$ and call
$$
\calA'(G, a, b, t-1, t-1, C^{[1]}, C^{[0]}, C^{[-1]}, C^{[-2]}).
$$

\item[$\boldsymbol{w\not\in\{a, a', b, b'\}, w'=b, b' > b > a' > w > a}$\ \ ] 
If $w$ wins at $V^u$ with $t$ points, this means that 
$\scr_2(w, V^{u}_{-1})=t$, $\scr_2(a', V^u_{-1})\le t-2$, $\scr_2(a, V^u_{-1})\le t-1$,
$\scr_2(c, V^u_{-1})\le t$ for all $c\in C\setminus\{w, a, a'\}$ with $w > c$, and
$\scr_2(c, V^u_{-1})\le t-1$ for all $c\in C\setminus\{w, a, a'\}$ with $c > w$.
Suppose that these conditions are satisfied. As $w$ still receives $t$ points at $V^v$, 
this means that $b$ wins at $V^v$ if and only if $\scr_2(b, V^{u}_{-1})=t-1$,
$\scr_2(b', V^{u}_{-1})\le t-2$.
Thus, we set $C^{[1]}=\emptyset$, $C^{[0]}=\{c\in C\setminus\{a\}\mid w > c\}$, 
$C^{[-1]}=\{c\in C\setminus\{a', b, b'\}\mid c > w\}$, $C^{[-2]}=\{a', b'\}$ and call
$$
\calA'(G, w, b, t, t-1, C^{[1]}, C^{[0]}, C^{[-1]}, C^{[-2]}).\qedhere
$$
\end{description}
\end{proof}

Just as for Plurality, we obtain the following corollary that describes the complexity 
of finding and testing level-2 strategies and improving strategies under 2-approval.

\begin{corollary}
Given a GS-game $G=(V, \calR_2, (A_i)_{i\in N(V, \calR_2)})$, where for each $i\in N$
the set $A_i$ consists of $i$'s truthful vote and a subset of his minimal manipulations, 
and a strategy $v_1'\in \calL(C)$ of player~1 we can decide in polynomial time
whether $v_1'$ is a level-2 strategy and whether $v'_1$ is an improving strategy. 
Moreover, we can decide in polynomial time whether player~1 has a level-2 strategy 
or an improving strategy in $G$.
\end{corollary}

We remark that the minimality assumption plays an important role in our analysis.
Indeed, in the absence of this assumption a promoter $i$ may manipulate by swapping
two different candidates (one of which is his most preferred 2-competitive candidate $p$)
into the top two positions. Let $v_i[\tp_2(v_i);\{p, c\}]$ be some such manipulation.
If we try to model this possibility via a network flow construction, we would have to add
edges from $i$ to both $p$ and $c$; the lower bounds on these edges would have to be set to 0, 
to allow $i$ to vote truthfully. However, there may then be a flow that uses
the edge $(i, c)$, but not $(i, p)$, which corresponds to 
a vote that promotes $c$, but not $p$; such a vote is not a level-1 strategy.

Interestingly, a level-2 player may want to swap two candidates into the top two positions, 
even if he assumes that all level-1 players use minimal strategies. In fact, the following example shows 
that a strategy of this form may weakly dominate all other non-equivalent strategies.

\begin{example}\label{ex:non-min}
{\em
Let the profile of sincere preferences be as in Table~\ref{table:non-min}, and assume that 
the voting rule is $2$-approval and the tie-breaking order is given by 
$a>b>c>d>\dots$. Assume that player 1 is the level-2 player.
The winner under $2$-approval is $a$ with two points; candidates $b$, $c$, and $d$
also have two points each. Voters 1, 4 and 5 are GS-manipulators; voter 4 may manipulate
by swapping $c$ into top two positions, and voter 5 may manipulate by swapping
$d$ into top two positions.

Consider the GS-game where $N=\{1, 4, 5\}$, $A_4=\{v_4, v_4[b;c]\}$, 
and $A_5=\{v_5, v_5[b;d]\}$ (note that both $A_4$ and $A_5$ only contain a proper subset
of the respective player's minimal manipulations; for instance, $v_4[u_1;c]\not\in A_4$). 
We claim that $v_1[\{u_5, u_6\};\{b, c\}]$
is a weakly dominant strategy for player~1. Indeed, consider the four possible scenarios:
\begin{itemize}
\item
Players 4 and 5 are truthful. Then the best outcome that voter 1 can ensure is that $b$ wins.
\item
Player 4 is truthful, but player 5 manipulates. Then the best outcome that voter 1 can ensure is that $c$ wins.
\item
Player 4 manipulates, but player 5 is truthful. Then the best outcome that voter 1 can ensure is that $c$ wins.
\item
Players 4 and 5 both manipulate. Then the best outcome that voter 1 can ensure is that $c$ wins.  
\end{itemize}
Now, it is clear that only votes that rank $b$ and $c$ in top two positions achieve 
all of these objectives simultaneously.

\begin{table}[h]
\begin{center}
\begin{tabular}{c|c|c|c|c|c|c}
voter 1&voter 2&voter 3&voter 4&voter 5&voter 6&voter 7\\
\midrule
$u_5$    &$a$      &$a$      &$b$      &$b$      &$c$      &$c$      \\
$u_6$    &$d$      &$d$      &$u_1$    &$u_2$    &$u_3$    &$u_4$    \\
$b$      &$\dots$  &$\dots$  &$c$      &$d$      &$a$      &$a$      \\
$c$      &$\dots$  &$\dots$  &$\dots$  &$\dots$  &$\dots$  &$\dots$  \\
$d$      &$\dots$  &$\dots$  &$\dots$  &$\dots$  &$\dots$  &$\dots$  \\
$\dots$  &$\dots$  &$\dots$  &$\dots$  &$\dots$  &$\dots$  &$\dots$  \\

\end{tabular}
\caption{The strategy $bc\dots$ of player 1 weakly dominates all non-equivalent strategies.
}\label{table:non-min}
\end{center}
\end{table}
}
\end{example}

\section{$\boldsymbol{k}$-Approval for $\boldsymbol{k\ge 3}$}\label{sec:hard}
Regrettably, our analysis of $k$-approval under the minimality assumption does not extend 
from $k=2$ to $k=3$. Specifically, the argument breaks down when we consider a potential demoter
under 3-approval who can only help his top candidate by swapping his second and third candidate 
out of the top three positions. If he chooses to manipulate, he has to perform both 
of these swaps at once; he can also remain truthful and not perform any swaps. 
It is not clear how to capture this all-or-nothing behavior via network flows.
We conjecture that finding a level-$2$ strategy under $3$-approval is computationally hard,  
even under minimality assumption. We will now prove a weaker result, showing that
this problem is NP-hard for $k$-approval with $k\ge 4$ (and without the minimality assumption).
Moreover, we will also show that it is coNP-hard to decide whether a given strategy is improving.

\begin{theorem}\label{thm:4app}
For every fixed $k\ge 4$, 
given a GS-game $G=(V, \calR_k, (A_i)_{i\in N})$ and a strategy $v$
of voter~1, it is {\em NP}-hard to decide whether $v$ is a level-2 strategy, 
and it is {\em coNP}-hard to decide whether $v$ is an improving strategy. 
\end{theorem}

\begin{proof}
Our hardness proof proceeds by a reduction from the classic NP-complete problem {\sc Exact Cover by 3-Sets (X3C)}.
An instance of this problem is given by a ground set $\Gamma = \{g_1, \dots,g_{3\nu}\}$ and
a collection $\Sigma = \{\sigma_1, \dots, \sigma_\mu\}$ of $3$-element subsets of $\Gamma$.
It is a `yes'-instance if there is a subcollection $\Sigma'\subseteq \Sigma$
with $|\Sigma'|=\nu$ such that $\cup_{\sigma\in \Sigma'}\sigma = \Gamma$, and a `no'-instance otherwise.

We will first establish that our problems are hard for $k=4$; towards the end of the proof, we will show
how to extend our argument to $k>4$.

Consider an instance $I^0 = (\Gamma^0, \Sigma^0)$ of {\sc X3C} with $|\Gamma|=3\nu'$. 
We will first modify this instance as follows. We add three new elements to $\Gamma^0$
and a set containing them to $\Sigma^0$. We then add $\nu'+2$ triples $x_i, y_i, z_i$, 
$i\in[\nu'+2]$, of new elements to $\Gamma^0$ and for each such triple we add 
the set $S_i= \{x_i, y_i, z_i\}$ to $\Sigma^0$. Finally, we add sets $S'_i = \{y_i, z_i, x_{i+1}\}$, 
$i\in[\nu'+1]$, and $S'_{\nu'+2} = \{y_{\nu'+2}, z_{\nu'+2}, x_1\}$ to $\Sigma^0$.
We then renumber the elements of the ground set
so that the elements added at the first step are numbered $g_1, g_2, g_3$.
We denote the resulting instance by $(\Gamma, \Sigma)$, and let $\nu=|\Gamma|/3$, $\mu=|\Sigma|$.
Clearly, $I = (\Gamma, \Sigma)$ is a `yes'-instance of X3C if and only if $I = (\Gamma^0, \Sigma^0)$ is.
We let $\widehat{\Sigma}=\{S_i, S'_i\mid i\in[\nu'+2]\}$; 
we have $\nu = 2\nu'+3$, $|\widehat{\Sigma}|=2(\nu'+2)$. 
 
In what follows, when writing $X \succ Y$ in the description of an order $\succ$,
we mean that all elements of $X$ are ranked above all elements of $Y$, but the order
of elements within $X$ and within $Y$ is not specified and can be arbitrary. 
We construct a GS-game as follows.
We introduce a set of candidates $C'=\{c_1, \dots, c_{3\nu}\}$
that correspond to elements of~$\Gamma$, 
three special candidates $w$, $p$, $c$, and, finally, a set of dummy candidates
$$
D = \bigcup_{i=0}^{\mu} D_i \cup D_c \cup \bigcup_{j=1}^{\nu+1} E_j \cup \bigcup_{i=1}^{3\nu}\bigcup_{j=1}^{\nu+1} F_{i, j}, 
$$
where $|D_i|=4$ for $i=0, \dots, \mu$, $|D_c|=3$, and $|E_j|=2$, $|F_{i, j}|=3$ for $i\in [3\nu]$, $j\in [\nu+1]$. 
Thus, the set of candidates is $C=\{w,p,c\}\cup C'\cup D$.
We define the tie-breaking  order $>$ on $C$ by setting 
$$
w>c>p>c_1>\dots>c_{3\nu}>D.
$$

For each $j\in[\mu]$, we let $C_j=\{c_i\mid g_i\in \sigma_j\}$.
The profile $V$ consists of $2+\mu+(3\nu+1)(\nu+1)$ votes defined as follows:
\begin{align*}
z_0 &= D_0\succ p\succ c_1\succ c  \succ C'\setminus\{c_1\} \succ \dots, \\
 z_i &= D_i\succ C_i \succ c \succ \dots, \hspace{4cm} i\in [\mu], \\
  u &= c  \succ D_c \succ w \succ\dots,\\
 u_{j} &= w\succ p \succ E_j \succ \dots, \hspace{4cm} j\in [\nu+1],\\
 u_{i, j} &= c_i \succ F_{i,j} \succ w \succ \dots \hspace{4cm} i \in [3\nu],\ j\in [\nu+1].
 \end{align*}
We have 
$$
\scr_4(w, V)=\scr_4(p, V)=\scr_4(c_i, V)=\nu+1 \quad \text{for all $i\in [3\nu]$},
$$ 
$\scr_4(c, V)=1$, and $\scr_4(d, V)\le 1$ for each $d\in D$.
Thus, $w$ wins under $4$-approval because of the tie-breaking rule.

We have $S(V,4)=C'\cup\{p\}$. The set of GS-manipulators in this profile consists
of the first $\mu+1$ voters; we assume that the first voter (i.e., voter 0)
is the level-2 voter. 
We now define a GS-game for this profile by constructing the players' sets of strategies as follows:
\[
z'_i = z_i[D_i; C_i\cup\{c\}], \quad A_i = \{z_i, z'_i\} \text{ for all $i\in[\mu]$.}
\]
Observe that for each $i\in [\mu]$ the vote
$z'_i$ is a level-1 strategy for voter $i$, which makes $i$'s top candidate in $C_i$
the winner with $\nu+2$ points (note that voter $i$ orders $C_i$ in the same way as $>$ does, 
so tie-breaking favours $i$'s most preferred candidate in $C_i$).
This completes the description of our game $G$.

Fix some $d,d'\in D_0$ and let 
$$
z'_0 = z_0[\{d, d'\}; \{p, c\}], \qquad 
z''_0=z_0[d;p].
$$
Note that both $z'_0$ and $z''_0$ are level-1 strategies for voter $0$, 
which make $p$ the winner with $\nu+2$ points. 
Clearly, we can construct the profile $V$ and the players' sets of strategies 
in polynomial time given $I$.\par\smallskip

We will now argue that $z'_0$ is an improving strategy if and only if $I=(\Gamma, \Sigma)$ is a `no'-instance of {\sc X3C}, 
and that $z''_0$ is a level-2 strategy if and only if $I=(\Gamma, \Sigma)$ is a `yes'-instance of {\sc X3C}.

As a preliminary observation, consider some strategy $z$ of voter $0$ such that $\tp_4(z)$ consists of $p$ 
and three dummy candidates. By construction, 
for every profile of other players' strategies, $z$ and $z_0''$ result in the same outcome.
Moreover, if everyone except voter $0$ votes truthfully, voter $0$ strictly prefers $z''_0$
to every strategy $\hat{z}$ with $\tp_4(\hat{z})\subseteq D$. Thus, $z''_0$
can only be weakly dominated by a strategy that places at least one candidate in $C'\cup \{c, w\}$ 
in top $4$ positions.

Suppose first that $I$ is a `yes'-instance of {\sc X3C}. Fix a subcollection $\Sigma'$
witnessing this, and consider a profile $V'$ where the GS-manipulators that correspond to sets in $\Sigma'$ 
vote strategically, whereas everyone else votes truthfully. 
We have $\scr_4(p, V')=\scr_4(c, V')=\scr_4(w, V') = 
\nu+1$, $\scr_4(c_i, V')=\nu+2$ for all $c_i\in C'$, 
so $c_1$ wins. 
However, if voter $0$ changes his vote to $z'_0$, the winner would be $c$, and voter $0$ prefers $c_1$ to $c$,
so voter $0$ strictly prefers voting $z_0$ over voting $z'_0$ in this case, 
i.e., $z'_0$ is not an improving strategy.

Now, if voter $0$ changes her vote to $z''_0$ instead, $p$ becomes the election winner, 
which is the best feasible outcome from voter $0$'s perspective. The only
way for voter $0$ to achieve this outcome is to rank $p$ and some dummy candidates in the top $4$ positions;
any vote $\widehat{z}$ with $\tp_4(\widehat{z})\cap (C'\cup\{c, w\})\neq\emptyset$ is strictly worse for voter $0$,
and hence cannot weakly dominate $z''_0$. As we have already observed that no strategy $\widehat{z}$ 
with $\tp_4(\widehat{z})\cap (C'\cup\{c, w\})=\emptyset$ can weakly dominate $z''_0$,
it follows that if $I$ is a `yes'-instance of X3C then $z''_0$ is a level-2 strategy.

On the other hand, suppose that $I$ is a `no'-instance of X3C.
Consider a strategy profile $V^*$ in $G$, and let 
$\Sigma''$ be a subcollection of $\Sigma$ that corresponds to players in $[\mu]$ 
who vote non-truthfully in $V^*$; we know that $\Sigma''$ is not an exact cover of $\Gamma$.
We will argue that voter $0$ weakly prefers $z'_0$ to both $z_0$ and $z''_0$ for every choice of $\Sigma''$,
and there are choices of $\Sigma''$ for which this preference is strict.

If $\Sigma''=\emptyset$, i.e., all voters in $[\mu]$ are truthful, then voter $0$ benefits
from changing his vote from $z_0$ to $z'_0$, as this vote makes $p$ the winner. 
Similarly, suppose that
all sets in $\Sigma''$ are pairwise disjoint (and hence $|\Sigma''| \le \nu-1$).
Then candidate $c$ gets at most $\nu$ points and
the winner in $V[V^*_{-0}, z_0]$ is one of the candidates from $C'$ (with $\nu+2$ points).
On the other hand, the winner in $V[V^*_{-0}, z'_0]$ 
is $p$ (with $\nu+2$ points), so voter $0$
benefits from changing his vote to $z'_0$.
In both of these cases, $z''_0$ has the same effect as $z'_0$.

Now, suppose that the sets in $\Sigma''$ are not pairwise disjoint. 
Let $X$ be the set of elements that appear in the largest
number of sets in $\Sigma''$, and let $g_\ell$ be the element of $X$
with the smallest index. Note that $g_\ell\neq g_1$, since we modified our instance of X3C so that 
$g_1$ only occurs in one set.
The winner in $V[V^*_{-0}, z_0]$ is either $c_\ell$ or $c$, and the winner's score is at least $\nu+3$.
Suppose that voter $0$ changes his vote from $z_0$ to $z'_0$.
If the winner in $V[V^*_{-0}, z_0]$ was $c$, this remains to be the case, 
and if the winner was $c_\ell$ then either $c_\ell$ remains the winner or $c$ becomes the winner, 
and voter $0$ prefers $c$ to $c_\ell$. Thus, in this case voting $z'_0$
is at least as good as voting $z_0$, and voting $z''_0$ has the same effect as voting $z_0$.

We conclude that whenever $\Sigma''$ is not an exact cover of $\Gamma$, voting $z'_0$ is at least
as good as voting $z_0$ or $z''_0$. 
It remains to establish that $z'_0$ is sometimes strictly better than either of these strategies.
To this end, suppose that $\Sigma''=\widehat{\Sigma}$. 
If voter $0$ votes $z'_0$, then the scores of the candidates covered by sets in $\widehat{\Sigma}$
are $\nu+3$, the score of $c$ is $1+2(\nu'+2) +1 =2\nu'+6=\nu+3$, and all other candidates have lower scores, so $c$ wins.
However, if voter $0$ votes $z_0$ or $z''_0$, $c$'s score is $\nu+2$, and therefore the winner is a candidate in $C'$.
Thus, in those circumstances, voter $0$ strictly prefers $z'_0$ to both $z_0$ and $z''_0$.
Hence, if $I$ is a `no'-instance of X3C, $z'_0$ weakly dominates $z_0$ and $z''_0$, 
and hence $z''_0$ is not a level-2 strategy.
This completes the proof for $k=4$. 

For $k>4$, we modify the construction by introducing $|V|$ additional groups of dummy candidates
$H_1, \dots, H_{|V|}$ of size $k-4$ each. We renumber the voters from $1$ to $|V|$ and modify 
the preferences of the $i$-th voter, $i\in [\mu]$, by inserting the 
group $H_i$ in positions $5, \dots, k$, and adding all other new dummy candidates
at the bottom of his ranking. Then the $k$-approval scores of all candidates in $C$
remain the same as in the original construction, and the $k$-approval score of each new dummy candidate
is $1$. The rest of the proof then goes through without change. 
\end{proof}

We note that the strategies of level-1 players in our hardness proof are not minimal;
determining whether our hardness result remains true under the minimality assumption 
is an interesting research challenge.

\begin{remark}
Our complexity lower bounds are not tight: we do not know whether the computational problems
we consider are in, respectively, {\em NP} and {\em coNP}. The following argument provides upper bounds 
on their complexity.

For every $n$-player game $G=(N, (A_i)_{i\in N}, (\succeq_i)_{i\in N})$, 
where each relation $\succeq_i$ is represented by a polynomial-time computable function of its arguments, and 
for every pair of strategies $u$, $v$ of player~1, the problem of deciding whether $u$ weakly dominates $v$ 
belongs to the complexity class {\em DP} (difference polynomial-time) \citep{PY84,Wechsung85}. 
Indeed, $u$ weakly dominates $v$ if and only if 
\begin{itemize}
\item[(a)] 
for every profile $P_{-1}$ of other players' strategies we have $(P_{-1}, u)\succeq_1 (P_{-1}, v)$
(which can be checked in {\em coNP}), and
\item[(b)] 
for some profile $P_{-1}$ of other players' strategies we have $(P_{-1}, u)\succ_1 (P_{-1}, v)$
(which can be checked in {\em NP}), 
\end{itemize}
i.e., the language associated with our problem is an intersection of an {\em NP}-language and a {\em coNP}-language.

Thus, for every GS-game based on a polynomial-time voting rule (including $k$-approval) the problem of checking whether a given
strategy is improving is in {\em DP}. This also means that
for $k$-approval with a fixed value of $k$ the problem
of checking whether a given strategy is a level-2 strategy belongs to the boolean hierarchy \citep{Wechsung85}, 
as there are only ${m\choose k}\le m^k$ pairwise non-equivalent votes, and 
it suffices to check that none of these votes weakly dominates the given strategy.
\end{remark}

\section{Conclusions and Further Research}\label{sec:conclusions}
We have initiated the analysis of voting games from the perspective of the cognitive hierarchy.
We have adopted a distribution-free approach that uses the concept of weak dominance in order
to reason about players' actions. The resulting framework is mathematically rich, captures some interesting behaviours, 
and presents a number of algorithmic challenges, even for simple voting rules. To illustrate this, we focused
on a well-known family of voting rules, namely, $k$-approval with $k\ge 1$, and investigated the complexity
of finding level-2 strategies and improving strategies under various rules in this family. For Plurality, i.e., for $k=1$, 
level-2 strategies and improving strategies are easy to find, and for $k\ge 4$ these problems are computationally hard, 
but for $k=2, 3$ we do not have a full understanding of the complexity of these problems. 
We identify a natural assumption (namely, the minimality assumption),
which is sufficient to obtain an efficient algorithm for $k=2$; however, it is 
not clear if it remains useful for larger values of $k$.

We list a few specific algorithmic questions that remain open:
\begin{itemize}
\item Is there a polynomial-time algorithm for computing level-2 strategies and improving strategies under $2$-approval 
without the minimality assumption?
\item Does Theorem~\ref{thm:4app} remain true under the minimality assumption?
\item What can be said about $3$-approval, with or without the minimality assumption?
\item What can be said about other prominent voting rules, most importantly the Borda rule?
\end{itemize}

In our analysis, we have focused on level-1 and level-2 voters. It would also be interesting
to extend our formal definitions to level-$\ell$ players for $\ell\ge 3$ and 
to investigate the associated algorithmic issues. 
While it is intuitively clear that the view of the game for these players will be more complex, 
it appears that for Plurality our algorithm can be extended in a straightforward manner; 
however, it is not clear if this is also the case for 2-approval.
Another interesting question, which can be analysed empirically, is whether truthful voting is likely 
to be a level-2 strategy, or, more broadly, how many votes in $\calL(C)$ are level-2 strategies; again, 
this question can also be asked for level-$\ell$ strategies with $\ell\ge 2$. 
Note that the analysis in our paper contributes algorithmic tools tools to tackle this issue.

A yet broader question, which can only be answered by combining empirical data and theoretical analysis,  
is whether the cognitive hierarchy approach provides a plausible
description of strategic behavior in voting. While our paper makes the first steps towards answering it, 
there is more to be done to obtain a full picture.

\bibliographystyle{abbrvnat}
\bibliography{voting}


\end{document}